\documentclass[a4paper,12pt]{article}
\usepackage{graphicx} 
\usepackage[utf8]{inputenc}
\usepackage{amsmath}
\usepackage{amsfonts}
\usepackage{amsthm}
\usepackage{amssymb}

\usepackage{enumerate}
\usepackage{url}
\usepackage{color}
\usepackage{ulem}
\usepackage {graphs}
\usepackage {graphics, graphicx}
\usepackage{stmaryrd}

\usepackage{pstricks}
\usepackage{pstricks-add}
\usepackage{pst-plot}

\newcommand{\Nc}{\mathcal{N}}
\newcommand{\Pc}{\mathcal{P}}

\newcommand{\Nb}{\mathbb{N}}

\newcommand{\delete}[1]{}
\newcommand\scalemath[2]{\scalebox{#1}{\mbox{\ensuremath{\displaystyle #2}}}}

\newtheorem{theorem}{Theorem}[section]

\newtheorem{corollary}[theorem]{Corollary}

\newenvironment{example}[1][Example]{\begin{trivlist}
\item[\hskip \labelsep {\bfseries #1}]}{\end{trivlist}}

\begin{document}

\title{On the Complexity of the Mis\`ere Version of Three Games Played on Graphs.}
\author{Gabriel Renault\thanks{Univ. Bordeaux, LaBRI, UMR5800, F-33400 Talence France \newline
CNRS, LaBRI, UMR5800, F-33400 Talence, France},
        Simon Schmidt 
\thanks{Corresponding author. Postal adress: Univ. Grenoble,  Institut Fourier, sfr Maths \`a Modeler, 100 rue des maths, BP 74, 38402 St Martin d'Heres cedex, France.
Phone number: +33 4 76 51 46 56. Email: simonschmidt(at)club-internet.fr }\thanks{Both authors are supported by PEPS Misere Grant.} }
\maketitle
\begin{abstract}
 We investigate the complexity of finding a winning strategy for the mis\`ere version of three games played on graphs: two variants of the game \textsc{NimG}, 
introduced by Stockmann in 2004 and the game \textsc{Vertex Geography} on both directed and undirected graphs. 
We show that on general graphs those three games are \textsc{pspace}-Hard or Complete.
For one \textsc{pspace}-Hard variant of \textsc{NimG}, we find an algorithm to compute an effective winning strategy in time $\mathcal{O}(\sqrt{|V(G)|}.|E(G)|)$ when $G$ is a bipartite 
graph. 
\begin{center} \textbf{Keywords.}  Combinatorial Games, Complexity, Graphs, Mis\`ere.\end{center}
\end{abstract}

\section*{Short author's biography.}

Gabriel Renault recently graduated his PhD entitled ``Combinatorial Games on Graphs". He worked under the supervision of Paul Dorbec and \'Eric Sopena in the LaBRI team at Bordeaux University.

Simon Schmidt is a PhD student under the supervision of Sylvain Gravier in the Maths \`a Modeler team at Joseph Fourier's University  in Grenoble. He works on combinatorial games played on graphs and on graph parameters defined by games.

\section{Background and definitions}

 We assume that the reader has some knowledge in combinatorial games theory. Basic definitions can be found in \cite{LIP}. 
We only remind that $o^+(G)$ denotes the normal outcome of the game $G$, whereas $o^-(G)$ denotes the mis\`ere outcome. 
The outcome of a game is $\Pc$ if the second player has a winning strategy and $\Nc$ if the first to move can win. Graph theoretical
notions used in this paper are standard and according to \cite{Beng-Jensen}. When it makes a difference to allow graphs to have loops, this will be pointed out. 
Complexity notions for games are those defined by Fraenkel in \cite{Fraenkel}.

 In this work we study the complexity of computing the mis\`ere outcome of three impartial combinatorial games played on graphs or directed graphs. 
Two of those games are variants of the famous game called Nim, which was solved by Bouton in 1901 \cite{Bouton}. In those variants, introduced by Stockman in \cite{stockman}, the heaps of tokens are placed on the vertices of 
a graph. Alternately, the players remove some tokens from the current heap and move along the edges of the graph. 
The order in which these two actions are done during a turn leads to two different games: 
\textsc{NimG-RM}, for ``Remove then Move'' and \textsc{NimG-MR}, for ``Move then Remove''. 
The game \textsc{NimG-RM} is played on a graph $G$ together with a function $w:V(G)\rightarrow \Nb$, 
called the weight function. For a vertex $u$, $w(u)$ represents the number of tokens on $u$. This game is played as follows:
\begin{itemize}
 \item There is a pointer on the starting vertex.
 \item The two players play alternately.
 \item During his turn, a player removes any number of tokens from the pointed vertex $u$, and then moves the pointer to a vertex $v$ in the neighbourhood
of $u$. At least one token must be removed.
 \item The player who starts his turn on a vertex with null weight loses in normal convention and wins in mis\`ere convention.
\end{itemize}

We denote by $(G,u,w)$ the game played on the graph $G$, with $u$ as starting vertex and $w$ as weight function. We also denote by $(u,k,v)$ the move consisting in decreasing
$w(u)$ to $k < w(u)$ and then moving to the vertex $v$. 

The game \textsc{NimG-MR} is exactly the same game as above, except that the player starts his turn by moving the pointer and then removes tokens from the pointed vertex.
If a player is forced to move to a null weight vertex, he loses in normal convention and wins in mis\`ere convention. 

We denote by $(G,u,w)$ the game played on the graph $G$, with $u$ as starting vertex and $w$ as weight function. We also denote by $(u,k,v)$ the move consisting 
in moving from the vertex $u$ to the vertex $v$ and then decreasing
$w(v)$ to $k < w(v)$.

\begin{example}
{\rm
Figure~\ref{fig:NGRM} gives an example of a move in \textsc{NimG-RM}.
The current vertex is grey.
The player whose turn it is chooses to remove two tokens from the current vertex and to move to the vertex with one token.
Figure~\ref{fig:NGMR} gives an example of a move in \textsc{NimG-MR}.
The current vertex is grey. 
The player whose turn it is starts by moving the current vertex to the vertex on its right. Then he removes all the tokens from this vertex.
}
\end{example}
\begin{figure}[h!]
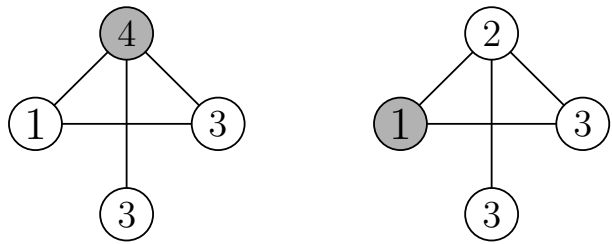

\begin{center}

\scalebox{1.2}{
\begin{graph}(9,3)

\graphnodesize{0.6}
\graphnodecolour{1}
\autodistance{1}
\roundnode{u}(1,1)
\autonodetext{u}{$\scalemath{1.2}{1}$}
\roundnode{a}(2,2)[\graphnodecolour{0.7}]
\autonodetext{a}{$\scalemath{1}{4}$}
\roundnode{b}(3,1)
\autonodetext{b}{$\scalemath{1}{3}$}
\roundnode{c}(2,0)
\autonodetext{c}{$\scalemath{1}{3}$}
\edge{u}{a}
\edge{u}{b}
\edge{a}{c}
\edge{a}{b}

\roundnode{u1}(5,1)[\graphnodecolour{0.7}]
\autonodetext{u1}{$\scalemath{1.2}{1}$}
\roundnode{a1}(6,2)
\autonodetext{a1}{$\scalemath{1}{2}$}
\roundnode{b1}(7,1)
\autonodetext{b1}{$\scalemath{1}{3}$}
\roundnode{c1}(6,0)
\autonodetext{c1}{$\scalemath{1}{3}$}
\edge{u1}{a1}
\edge{u1}{b1}
\edge{a1}{c1}
\edge{a1}{b1}

\end{graph}
}
\vspace{0.3cm}
\caption{Playing a move in \textsc{NimG-RM}}\label{fig:NGRM}
\end{center}
\end{figure}
\begin{figure}[h!]
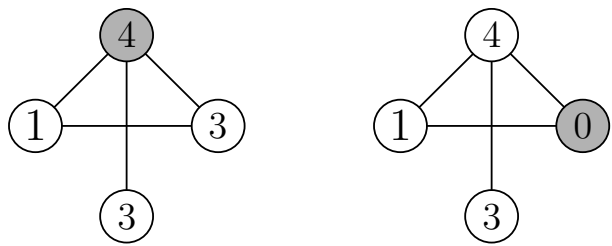

\begin{center}

\scalebox{1.2}{
\begin{graph}(9,3)

\graphnodesize{0.6}
\graphnodecolour{1}
\autodistance{1}
\roundnode{u}(1,1)
\autonodetext{u}{$\scalemath{1.2}{1}$}
\roundnode{a}(2,2)[\graphnodecolour{0.7}]
\autonodetext{a}{$\scalemath{1}{4}$}
\roundnode{b}(3,1)
\autonodetext{b}{$\scalemath{1}{3}$}
\roundnode{c}(2,0)
\autonodetext{c}{$\scalemath{1}{3}$}
\edge{u}{a}
\edge{u}{b}
\edge{a}{c}
\edge{a}{b}

\roundnode{u1}(5,1)
\autonodetext{u1}{$\scalemath{1.2}{1}$}
\roundnode{a1}(6,2)
\autonodetext{a1}{$\scalemath{1}{4}$}
\roundnode{b1}(7,1)[\graphnodecolour{0.7}]
\autonodetext{b1}{$\scalemath{1}{0}$}
\roundnode{c1}(6,0)
\autonodetext{c1}{$\scalemath{1}{3}$}
\edge{u1}{a1}
\edge{u1}{b1}
\edge{a1}{c1}
\edge{a1}{b1}

\end{graph}
}
\vspace{0.3cm}
\caption{Playing a move in \textsc{NimG-MR}}\label{fig:NGMR}
\end{center}
\end{figure}

The third game we focus on is called \textsc{Geography}. \textsc{Geography} is an impartial game played on a directed graph with a token on a vertex.
There exist two variants of the game: \textsc{Vertex Geography} and \textsc{Edge Geography}.
A move in \textsc{Vertex Geography} is to slide the token through an arc and delete the vertex on which the token was.
A move in \textsc{Edge Geography} is to slide the token through an arc and delete the edge on which the token just slid.
In both variants, the game ends when the token is on a sink.

A position is described by a graph and a vertex indicating where the token is.
\begin{figure}[h!]
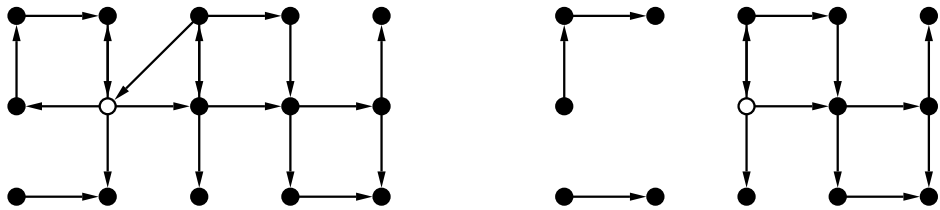

\begin{center}

\scalebox{1.2}{
\begin{graph}(10,3.5)(0,0)
\graphlinewidth{0.025}
\grapharrowlength{0.18}
\roundnode{a}(0,0)
\roundnode{b}(1,0)
\roundnode{c}(2,0)
\roundnode{d}(3,0)
\roundnode{e}(4,0)
\roundnode{f}(0,1)
\roundnode{g}(1,1)[\graphnodecolour{1}]
\roundnode{h}(2,1)
\roundnode{i}(3,1)
\roundnode{j}(4,1)
\roundnode{k}(0,2)
\roundnode{l}(1,2)
\roundnode{m}(2,2)
\roundnode{n}(3,2)
\roundnode{o}(4,2)

\diredge{a}{b}
\diredge{d}{e}
\diredge{f}{k}
\diredge{g}{b}
\diredge{g}{f}
\diredge{g}{h}
\diredge{g}{l}
\diredge{h}{c}
\diredge{h}{i}
\diredge{h}{m}
\diredge{i}{d}
\diredge{i}{j}
\diredge{j}{e}
\diredge{j}{o}
\diredge{k}{l}
\diredge{l}{g}
\diredge{m}{g}
\diredge{m}{h}
\diredge{m}{n}
\diredge{n}{i}

\roundnode{a0}(6,0)
\roundnode{b0}(7,0)
\roundnode{c0}(8,0)
\roundnode{d0}(9,0)
\roundnode{e0}(10,0)
\roundnode{f0}(6,1)
\roundnode{h0}(8,1)[\graphnodecolour{1}]
\roundnode{i0}(9,1)
\roundnode{j0}(10,1)
\roundnode{k0}(6,2)
\roundnode{l0}(7,2)
\roundnode{m0}(8,2)
\roundnode{n0}(9,2)
\roundnode{o0}(10,2)

\diredge{a0}{b0}
\diredge{d0}{e0}
\diredge{f0}{k0}
\diredge{h0}{c0}
\diredge{h0}{i0}
\diredge{h0}{m0}
\diredge{i0}{d0}
\diredge{i0}{j0}
\diredge{j0}{e0}
\diredge{j0}{o0}
\diredge{k0}{l0}
\diredge{m0}{h0}
\diredge{m0}{n0}
\diredge{n0}{i0}

\end{graph}
}
\vspace{0.3cm}
\caption{Playing a move in \textsc{Vertex Geography}}\label{fig:vergeo}
\end{center}
\end{figure}
\begin{example}
{\rm
Figure~\ref{fig:vergeo} gives an example of a move in \textsc{Vertex Geography}.
The token is on the white vertex.
The player whose turn it is chooses to move the token through the arc to the right.
After the removing of this vertex, some vertices (on the left of the directed graph) are no longer reachable.
Figure~\ref{fig:edgeo} gives an example of a move in \textsc{Edge Geography}.
The token is on the white vertex.
The player whose turn it is chooses to move the token through the arc to the right.
After that move, it is possible to go back to the previous vertex immediately as the arc in the other direction is still in the game.
}
\end{example}
\begin{figure}
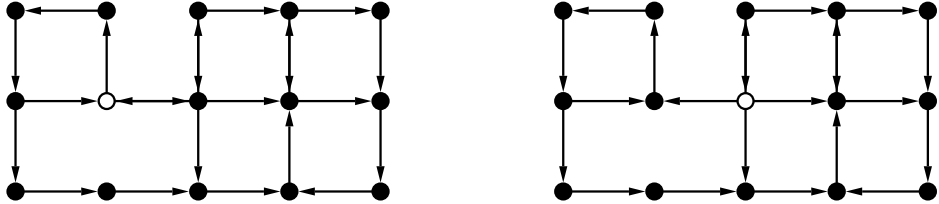

\begin{center}

\scalebox{1.2}{
\begin{graph}(10,3.5)(0,0)
\graphlinewidth{0.025}
\grapharrowlength{0.18}
\roundnode{a}(0,0)
\roundnode{b}(1,0)
\roundnode{c}(2,0)
\roundnode{d}(3,0)
\roundnode{e}(4,0)
\roundnode{f}(0,1)
\roundnode{g}(1,1)[\graphnodecolour{1}]
\roundnode{h}(2,1)
\roundnode{i}(3,1)
\roundnode{j}(4,1)
\roundnode{k}(0,2)
\roundnode{l}(1,2)
\roundnode{m}(2,2)
\roundnode{n}(3,2)
\roundnode{o}(4,2)

\diredge{a}{b}
\diredge{b}{c}
\diredge{c}{d}
\diredge{d}{i}
\diredge{e}{d}
\diredge{f}{a}
\diredge{f}{g}
\diredge{g}{h}
\diredge{g}{l}
\diredge{h}{c}
\diredge{h}{g}
\diredge{h}{i}
\diredge{h}{m}
\diredge{i}{j}
\diredge{i}{n}
\diredge{j}{e}
\diredge{k}{f}
\diredge{l}{k}
\diredge{m}{h}
\diredge{m}{n}
\diredge{n}{i}
\diredge{n}{o}
\diredge{o}{j}

\roundnode{a0}(6,0)
\roundnode{b0}(7,0)
\roundnode{c0}(8,0)
\roundnode{d0}(9,0)
\roundnode{e0}(10,0)
\roundnode{f0}(6,1)
\roundnode{g0}(7,1)
\roundnode{h0}(8,1)[\graphnodecolour{1}]
\roundnode{i0}(9,1)
\roundnode{j0}(10,1)
\roundnode{k0}(6,2)
\roundnode{l0}(7,2)
\roundnode{m0}(8,2)
\roundnode{n0}(9,2)
\roundnode{o0}(10,2)

\diredge{a0}{b0}
\diredge{b0}{c0}
\diredge{c0}{d0}
\diredge{d0}{i0}
\diredge{e0}{d0}
\diredge{f0}{a0}
\diredge{f0}{g0}
\diredge{g0}{l0}
\diredge{h0}{c0}
\diredge{h0}{g0}
\diredge{h0}{i0}
\diredge{h0}{m0}
\diredge{i0}{j0}
\diredge{i0}{n0}
\diredge{j0}{e0}
\diredge{k0}{f0}
\diredge{l0}{k0}
\diredge{m0}{h0}
\diredge{m0}{n0}
\diredge{n0}{i0}
\diredge{n0}{o0}
\diredge{o0}{j0}

\end{graph}
}
\vspace{0.3cm}
\caption{Playing a move in \textsc{Edge Geography}}\label{fig:edgeo}
\end{center}
\end{figure}

\textsc{Geography} can also be played on an undirected graph $G$ by seeing it as a symmetric directed graph where the vertex set remains the same and the arc set is $\{(u,v),(v,u)|(u,v) \in E(G)\}$, except that in the case of \textsc{Edge Geography}, going through an edge $(u,v)$ would remove both the arc $(u,v)$ and the arc $(v,u)$ of the directed version, to leave an undirected graph.

A \textsc{Geography} position is denoted $(G,u)$ where $G$ is the graph, or the directed graph, on which the game is played, and $u$ is the vertex of $G$ where the token is.

 The complexity of computing the normal outcome of these three games was already known. 
Burke and George\cite{Burke} proved that the game \textsc{NimG-MR} is \textsc{pspace}-Hard in normal convention, whereas Duch\^ene and Renault\cite{Renault}
found that \textsc{NimG-RM} is solvable in polynomial time. Lichtenstein and Sipser\cite{Lichtenstein} proved that finding the normal outcome of a \textsc{Vertex Geography} position on a directed graph 
is \textsc{pspace}-complete.
Schaefer\cite{Schaefer} proved that finding the normal outcome of an \textsc{Edge Geography} position on a directed graph is \textsc{pspace}-complete.
On the other hand, Fraenkel, Scheinerman and Ullman\cite{FraenkelGeo} gave a polynomial-time algorithm for finding the normal outcome of a \textsc{Vertex Geography} 
position on any undirected graph, and they also proved that finding the normal outcome of an \textsc{Edge Geography} position on an undirected graph 
is \textsc{pspace}-complete.

In this paper, we extend the investigation to
their mis\`ere version. The second section is devoted to \textsc{Geography}, and the third section to \textsc{NimG}. 

\section{Complexity results for \textsc{Geography} in mis\`ere convention}\label{sectionGeo}

We look here at the game \textsc{Geography} under mis\`ere convention, and show the problem is \textsc{pspace}-complete both on directed graphs and on undirected graphs, 
for both \textsc{Vertex Geography} and \textsc{Edge Geography}. 

We recall all the results in the table below. The stars indicate the results we show here.

\begin{table}[!h]\label{tabcomplNimGRD}
\centering
\begin{tabular}{|l|*{2}{c|}}
 \hline
~      & \textsc{Edge Geography} & \textsc{Vertex Geography}    \\
\hline
Normal & \textsc{pspace}-complete   & \textsc{pspace}-complete    \\
\hline
Misere & \textsc{pspace}-complete (*) & \textsc{pspace}-complete (*)\\
\hline  
\end{tabular}
\caption{ Complexity of \textsc{Geography} on directed graph.}
~\\
\begin{tabular}{|l|*{2}{c|}}
\hline
~      & \textsc{Edge Geography} & \textsc{Vertex Geography} \\
\hline
Normal & \textsc{pspace}-complete & Polynomial \\
\hline
Misere & \textsc{pspace}-complete (*) & \textsc{pspace}-complete (*) \\
\hline  
\end{tabular}
\caption{Complexity of \textsc{Geography} on undirected graph.}
\end{table}

First note that all these problems are in \textsc{pspace} as the length of a game of \textsc{Vertex Geography} is bounded by the number of its vertices, and the length of a game of \textsc{Edge Geography} is bounded by the number of its edges.

We start with \textsc{Vertex Geography} on directed graphs, where the reduction is quite natural, we just add a losing move to every position of the previous graph, move that the players will avoid until it becomes the only available move, that is when the original game will over.

\begin{theorem}
\label{thm:dirgeo}
Finding the mis\`ere outcome of a \textsc{Vertex Geography} position on a directed graph is \textsc{pspace}-complete.
\end{theorem}
\begin{proof}
We reduce the problem from normal \textsc{Vertex Geography} on directed graphs.

Let $G$ be a directed graph.
Let $G'$ be the directed graph with vertex set
$$V(G') = \{u_1,u_2|u\in V(G)\}$$
and arc set
$$A(G') = \{(u_1,v_1)|(u,v)\in A(G)\}\cup\{(u_1,u_2)|u\in V(G)\}$$
that is the graph where each vertex of $G$ gets one extra out-neighbour that was not originally in the graph. 
We claim that the normal outcome of $(G,v)$ is the same as the mis\`ere outcome of $(G',v_1)$ and show it by induction on the number of vertices in $G$.

If $V(G)=\{v\}$, then both $(G,v)$ and $(G',v_1)$ are $\Pc$-positions.
Assume now $|V(G)| \geqslant 2$.
Assume first $(G,v)$ is an $\Nc$-position.
There is a winning move in $(G,v)$ to $(\widetilde{G},u)$.
We show that moving from $(G',v_1)$ to $(\widehat{G}',u_1)$ is a winning move.
We have $V(\widehat{G}') = V(\widetilde{G}')\cup \{v_2\}$ and $A(\widehat{G}') = A(\widetilde{G}')$.
As the vertex $v_2$ is disconnected from the vertex $u_1$ in $\widehat{G}'$, the games $(\widehat{G}',u_1)$ and $(\widetilde{G}',u_1)$ share the same game tree, and they both have outcome $\Pc$ by induction.
Hence $(G',v_1)$ has mis\`ere outcome $\Nc$.
Now assume $(G,v)$ is a \mbox{$\Pc$-position}.
For the same reason as above, moving from $(G',v_1)$ to any $(\widehat{G}',u_1)$ would leave a game whose mis\`ere outcome is the same as the normal outcome of a game obtained after playing a move in $(G,v)$, which is $\Nc$.
The only other available move is from $(G',v_1)$ to $(\widehat{G}',v_2)$, which is a losing move as it ends the game.
Hence $(G',v_1)$ has mis\`ere outcome $\Pc$.
\end{proof}

The proof in \cite{Lichtenstein} actually works even if we only consider planar bipartite directed graphs with maximum degree $3$.
As our reduction keeps the planarity and the bipartition, only adds vertices of degree $1$ and increases the degree of vertices by $1$, we get the following corollary.

\begin{corollary}
\label{cor:dirgeo}
Finding the mis\`ere outcome of a \textsc{Vertex Geography} position on a planar bipartite directed graph with maximum degree $4$ is \textsc{pspace}-complete.
\end{corollary}

For undirected graphs, adding a new neighbour to each vertex would work the same, but the normal version of \textsc{Vertex Geography} on undirected graphs is 
solvable in polynomial time, so we make a reduction from directed graphs, and replace each arc by an undirected gadget.
That gadget would need to act like an arc, that is a player who would want to take it in the wrong direction would lose the game, 
as well as a player who would want to take it when the vertex at the other end has already been played. We want also to force that the player 
who moves in the gadget is the same as the one who moves the token to the other end. In that way, it will be the other player's turn when the token reaches the end vertex of 
the arc gadget, as in the original game.

\begin{theorem}
\label{thm:ungeo}
Finding the mis\`ere outcome of a \textsc{Vertex Geography} position on an undirected graph is \textsc{pspace}-complete.
\end{theorem}
\begin{proof}
We reduce the problem from normal \textsc{Vertex Geography} on directed graphs.

We introduce a gadget that will replace any arc $(u,v)$ of the original directed graph, and add a neighbour to each vertex to have an undirected graph whose mis\`ere outcome is the normal outcome of the original directed graph.

Let $G$ be a directed graph.
Let $G'$ be the undirected graph with vertex set 
$$\begin{array}{r@{\ }c@{\ }l}
V(G') = & & \{u,u'|u \in V(G)\} \\
        & \cup & \{uv_i |(u,v) \in A(G), 1\leq i \leq 8\} \\
\end{array}$$
and edge set 
$$\begin{array}{r@{\ }c@{\ }l}
E(G') & = & \{(u,uv_1),(uv_1,uv_2),(uv_1,uv_3),(uv_1,uv_6),(uv_2,uv_4),(uv_3,uv_5), \\
        & & (uv_3,uv_6),(uv_4,uv_5),(uv_4,uv_6),(uv_5,uv_6),(uv_6,uv_7),(uv_7,uv_8), \\
        & & (uv_7,v) | (u,v) \in A(G)\} \\
      & \cup & \{(u,u')|u \in V(G)\}
\end{array}$$
that is the graph where every arc $(u,v)$ of $G$ has been replaced by the gadget of Figure~\ref{fig:diredge}, identifying both $u$ vertices and both $v$ vertices, and each vertex of $G$ gets one extra neighbour that was not originally in the graph.
\begin{figure}
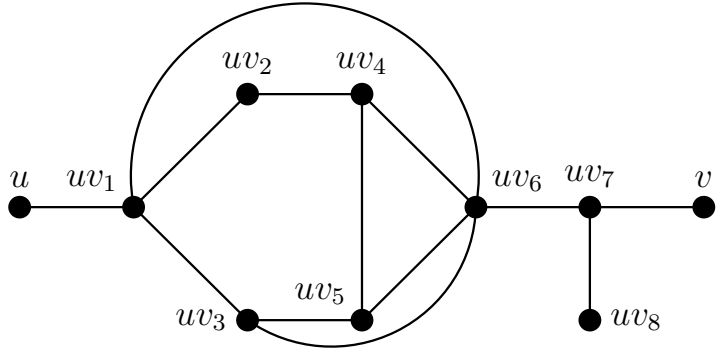

\begin{center}

\scalemath{1.5}{
\begin{graph}(6,3)(0,0)
\roundnode{a}(0,1)
\autonodetext{a}[n]{$\scalemath{0.75}{u}$}
\roundnode{b}(1,1)
\autonodetext{b}[nw]{$\scalemath{0.75}{uv_1}$}
\roundnode{c}(2,2)
\autonodetext{c}[n]{$\scalemath{0.75}{uv_2}$}
\roundnode{d}(2,0)
\autonodetext{d}[w]{$\scalemath{0.75}{uv_3}$}
\roundnode{e}(3,2)
\autonodetext{e}[n]{$\scalemath{0.75}{uv_4}$}
\roundnode{f}(3,0)
\autonodetext{f}[nw]{$\scalemath{0.75}{uv_5}$}
\roundnode{g}(4,1)
\autonodetext{g}[ne]{$\scalemath{0.75}{uv_6}$}
\roundnode{h}(5,1)
\autonodetext{h}[n]{$\scalemath{0.75}{uv_7}$}
\roundnode{i}(5,0)
\autonodetext{i}[e]{$\scalemath{0.75}{uv_8}$}
\roundnode{j}(6,1)
\autonodetext{j}[n]{$\scalemath{0.75}{v}$}

\edge{a}{b}
\edge{b}{c}
\edge{b}{d}
\bow{b}{g}{0.6}
\edge{c}{e}
\edge{d}{f}
\bow{d}{g}{-0.3}
\edge{e}{f}
\edge{e}{g}
\edge{f}{g}
\edge{g}{h}
\edge{h}{i}
\edge{h}{j}

\end{graph}
}
\vspace{0.6cm}
\caption{The arc gadget}\label{fig:diredge}
\end{center}
\end{figure}
We claim that the normal outcome of $(G,u)$ is the same as the mis\`ere outcome of $(G',u)$ and show it by induction on the number of vertices in $G$.

If $V(G) = u$, then $(G,u)$ is a normal $\Pc$-position.
In $(G',u)$ the first player can only move to $(\widehat{G}',u')$ where the second player wins as he cannot move.

Now assume $|V(G)| \geqslant 2$.

We first show that no player wants to move the token from $v$ to any $wv_7$, whether $w$ has been played or not.
We will only prove it for moving the token from $v$ to some $wv_7$ where $w$ is still in the game, as the other case is similar.
First note that, if $w$ is removed from the game in the sequence of moves following that first move, as $v$ is already removed, all vertices of the form $wv_i$ would be disconnected from the token, and therefore unreachable.
Hence whether the move from $wv_1$ to $w$ is winning does not depend on the set of vertices deleted in that sequence, and it is possible to argue the two cases.
Assume the first player moved the token from $v$ to any $wv_7$.
Then the second player can move the token to $wv_6$.
From there, the first player has four choices.
If she goes to $wv_1$, the second player answers to $wv_2$, then the rest of the game is forced and the second player wins.
If she goes to $wv_4$, he answers to $wv_2$ where she can only move to $wv_1$, and let him go to $wv_3$ where she is forced to play to $wv_5$ and she loses.
The case where she goes to $wv_5$ is similar.
In the case where she goes to $wv_3$, we argue two cases:
if the move from $wv_1$ to $w$ is winning, he answers to $wv_5$, where all is forced until he gets the move to $w$;
if that move is losing, he answers to $wv_1$, from where she can either go to $w$, which is a losing move by assumption, or go to $wv_2$ where every move is forced until she loses.

We now show that no player wants to move the token from $v$ to any $vw_1$ where $w$ has already been played.
Assume the first player just played that move.
Then the second player can move the token to $vw_3$.
From there, the first player has two choices.
If she plays to $vw_6$, he answers to $vw_4$, where she can only end the game and lose.
If she plays to $vw_5$, he answers to $vw_4$, where the move to $vw_2$ is immediately losing, and the move to $vw_6$ forces the token to go to $vw_7$ and then to $vw_8$, where she loses.

Assume first that $(G,u)$ is an $\Nc$-position.
There is a winning move in $(G,u)$ to some $(\widetilde{G},v)$.
We show that moving the token from $u$ to $uv_1$ in $G'$ is a winning move for the first player.
From there, the second player has three choices.
If he moves the token to $uv_6$, the first player answers to $uv_3$, then the rest of the game is forced and the first player wins.
If he moves the token to $uv_2$, the first player answers to $uv_4$, where the second player again has two choices:
either he goes to $uv_6$, she answers to $uv_5$ where he is forced to lose by going to $uv_3$;
or he goes to $uv_5$, she answers to $uv_6$ where the move to $uv_3$ is immediately losing and the move to $uv_7$ is answered to a game $(\widehat{G}',v)$.
As $u'$ and all vertices of the form $uv_i$ are either played or disconnected from $v$ in $\widehat{G}'$, the only differences in the possible moves in (followers of) the games $(\widehat{G}',v)$ and $(\widetilde{G}',v)$ are moves from a vertex $w$ to $wu_1$ or to $uw_7$, so they both have outcome $\Pc$ by induction.
The case where he chooses to move the token to $uv_3$ is similar.
Hence $(G',u)$ is an $\Nc$-position.

Now assume $(G,u)$ is a $\Pc$-position.
Then any $(\widetilde{G},v)$ that can be obtained after a move from $(G,u)$ is an $\Nc$-position.
Moving the token to $u'$ in $G'$ is immediately losing, so we may assume the first player moves it to some $uv_1$, where the second player answers to $uv_3$.
From there the first player has two choices.
If she goes to $uv_6$, the second player answers by going to $uv_4$, where both available moves are immediately losing.
If she goes to $uv_5$, he answers to $uv_4$, where the move to $uv_2$ is immediately losing, and the move to $uv_6$ is answered to $uv_7$, where again the move to $uv_8$ is immediately losing, so we may assume he moves the token to $v$.
As $u'$ and all vertices of the form $uv_i$ are either played or disconnected from $v$ in $\widehat{G}'$, the only differences in the possible moves in (followers of) the games $(\widehat{G}',v)$ and $(\widetilde{G}',v)$ are moves from a vertex $w$ to $wu_1$ or to $uw_7$, so they both have outcome $\Nc$ by induction.
Hence $(G',u)$ is a $\Pc$-position.
\end{proof}

Again, using the fact that the proof in \cite{Lichtenstein} actually works even if we only consider planar bipartite directed graphs with maximum degree $3$,
as our reduction keeps the planarity since the gadget is planar with the vertices we link to the rest of the graph being on the same face, only adds vertices of 
degree at most $5$ and increases the degree of vertices by $1$, we get the following corollary.

\begin{corollary}
\label{cor:undirgeo}
Finding the mis\`ere outcome of a \textsc{Vertex Geography} position on a planar undirected graph with degree at most $5$ is \textsc{pspace}-complete.
\end{corollary}

Though mis\`ere play is generally considered harder to solve than normal play, the feature that makes it hard is the fact that disjunctive sums do not behave as nicely as in normal play, and \textsc{Geography} is a game that does not split into sums.
Hence the above result appears a bit surprising as it was not expected.

We now look at \textsc{Edge Geography} where the reductions are very similar to the one for \textsc{Vertex Geography} on directed graphs.

We start with the undirected version.

\begin{theorem}
\label{thm:undiregeo}
Finding the mis\`ere outcome of an \textsc{Edge Geography} position on an undirected graph is \textsc{pspace}-complete.
\end{theorem}

\begin{proof}
We reduce the problem from normal \textsc{Edge Geography} on undirected graphs.

Let $G$ be an undirected graph.
Let $G'$ be the undirected graph with vertex set
$$V(G') = \{u_1,u_2|u\in V(G)\}$$
and edge set
$$E(G') = \{(u_1,v_1)|(u,v)\in E(G)\}\cup\{(u_1,u_2)|u\in V(G)\}$$
that is the graph where each vertex of $G$ gets one extra neighbour that was not originally in the graph. 
We claim that the normal outcome of $(G,v)$ is the same as the mis\`ere outcome of $(G',v_1)$ and show it by induction on the number of vertices in $G$.
The proof is similar to the proof of Theorem~\ref{thm:dirgeo}
\end{proof}

We now look at \textsc{Edge Geography} on directed graphs.

\begin{theorem}
\label{thm:diregeo}
Finding the mis\`ere outcome of an \textsc{Edge Geography} position on a directed graph is \textsc{pspace}-complete.
\end{theorem}

\begin{proof}
We reduce the problem from normal \textsc{Edge Geography} on directed graphs.

Let $G$ be a directed graph.
Let $G'$ be the directed graph with vertex set
$$V(G') = \{u_1,u_2|u\in V(G)\}$$
and arc set
$$A(G') = \{(u_1,v_1)|(u,v)\in A(G)\}\cup\{(u_1,u_2)|u\in V(G)\}$$
that is the graph where each vertex of $G$ gets one extra out-neighbour that was not originally in the graph. 
We claim that the normal outcome of $(G,v)$ is the same as the mis\`ere outcome of $(G',v_1)$ and show it by induction on the number of vertices in $G$.
The proof is similar to the proof of Theorem~\ref{thm:dirgeo}
\end{proof}

\section{Complexity results for \textsc{NimG} in mis\`ere convention}\label{sectionNimG}

In this section, we answer a question from Duchene and Renault. In \cite{Renault}, they found a polynomial algorithm to compute the normal outcome of 
the game \textsc{NimG-RM}
and asked if there is one in mis\`ere convention. We will show that the mis\`ere version of \textsc{NimG-RM} is \textsc{pspace}-Hard on general graphs. 
Our proof, like Burke and George's proof, used a reduction from the game \textsc{Vertex Geography}, which is \textsc{pspace}-Complete (see section \ref{sectionGeo}).
But if we only consider the game on bipartite graphs, we get an algorithm to find an effective strategy in time $\mathcal{O}(\sqrt{|V(G)|}.|E(G)|)$. 
We also show that \textsc{NimG-MR} is \textsc{pspace}-Hard in mis\`ere convention.

We start with a summary of the known results in the two tables below. The stars indicate the new results we prove in this paper. Because loops may sometimes make a difference, 
we note +L when there is a loop on all the vertices, and +NL when loops are not permitted. As said in the introduction, the results for the polynomial complexity
 of \textsc{NimG-RM} in normal play 
  are due to Renault and Duchene \cite{Renault}, whereas the \textsc{pspace}-Hardness results for \textsc{NimG-MR} are due to Burke 
and George \cite{Burke}.

\begin{table}[!h]\label{tabcomplNimGRD2}
\centering
\begin{tabular}{|l|*{3}{c|}}
 \hline
~      & \textsc{NimG-RM+L} & \textsc{NimG-RM}    & \textsc{NimG-RM+NL} \\
\hline
Normal & Polynomial   & Polynomial     & Polynomial        \\
\hline
Mis\`ere & Polynomial (*)       & \textsc{pspace}-Hard (*) & \textsc{pspace}-Hard (*)\\
\hline  
\end{tabular}
\caption{ Complexity of \textsc{NimG-RM}.}
~\\
\begin{tabular}{|l|*{3}{c|}}
\hline
~      & \textsc{NimG-MR+L} & \textsc{NimG-MR} & \textsc{NimG-MR+NL} \\
\hline
Normal & \textsc{pspace}-Hard & \textsc{pspace}-Hard & \textsc{pspace}-Hard \\
\hline
Mis\`ere & \textsc{pspace}-Hard (*) & \textsc{pspace}-Hard (*)& \textsc{pspace}-Hard (*) \\
\hline  
\end{tabular}
\caption{Complexity of \textsc{NimG-MR}.}
\end{table}

We start with the proof that the mis\`ere version of \textsc{NimG-RM+NL} is \textsc{pspace}-Hard on general graphs. 
We reduce the normal version of \textsc{Vertex Geography} on directed graphs to the mis\`ere version of \textsc{NimG-RM+NL}. 
When a vertex with only one token is visited in \textsc{NimG-RM}, its weight is necessarily decreased to $0$. 
Since in mis\`ere convention, moving to a null weight vertex is a losing move, no player
wants to move further to this vertex. Decreasing the weight function to $0$ in \textsc{NimG-RM+NL} is therefore convenient to simulate the clear of a vertex in 
\textsc{Vertex Geography}. Like in theorem \ref{thm:undiregeo}, the key of the proof is the design of a gadget that acts like an oriented edge.

\begin{theorem}
 The mis\`ere version of \textsc{NimG-RM+NL} is \textsc{pspace}-Hard. 
\end{theorem}
\begin{proof}

We perform our reduction as follows. Let $G$ be a directed graph standing for an instance of \textsc{Vertex geography}. 
We construct an undirected graph $G'$ and a weight function $w_{G'}$ as follows:
\begin{itemize}
 \item If $u\in V(G)$, $X_u$ is a vertex of $G'$.
 \item If $u,v\in V(G)$ and $(u,v)\in E(G)$, then $a_{uv}$, $b_{uv}$, $c_{uv}$, $d_{uv}$ are vertices of $G'$.
 \item If $(u,v)\in E(G)$, then $(X_u,a_{uv})$, $(a_{uv},b_{uv})$, $(b_{uv},c_{uv})$, $(b_{uv},d_{uv})$, $(c_{uv},d_{uv})$ and $(d_{uv},X_v)$ are undirected edges of $G'$.
\item The weight function is defined by $w_{G'}(X_u)=1$ and (forgetting the index) $w_{G'}(a)=w_{G'}(b)=w_{G'}(c)=1$ and $w_{G'}(d)=2$. 
\end{itemize}
This means we replace all the arcs $(u,v)$ by the gadget of figure \ref{fig:diredgeNIMGRM}. 
\begin{figure}
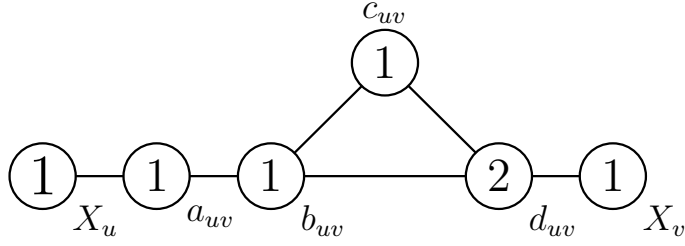

\begin{center}

\scalemath{1.5}{
\begin{graph}(5,1)(0,0)
\graphnodesize{0.6}
\graphnodecolour{1}
\autodistance{1}
\roundnode{u}(0,0)
\autonodetext{u}[se]{$\scalemath{0.75}{X_u}$}
\autonodetext{u}{$\scalemath{1.2}{1}$}
\roundnode{a}(1,0)
\autonodetext{a}[se]{$\scalemath{0.75}{a_{uv}}$}
\autonodetext{a}{$\scalemath{1}{1}$}
\roundnode{b}(2,0)
\autonodetext{b}[se]{$\scalemath{0.75}{b_{uv}}$}
\autonodetext{b}{$\scalemath{1}{1}$}
\roundnode{c}(3,1)
\autonodetext{c}[n]{$\scalemath{0.75}{c_{uv}}$}
\autonodetext{c}{$\scalemath{1}{1}$}
\roundnode{d}(4,0)
\autonodetext{d}[se]{$\scalemath{0.75}{d_{uv}}$}
\autonodetext{d}{$\scalemath{1}{2}$}
\roundnode{v}(5,0)
\autonodetext{v}[se]{$\scalemath{0.75}{X_v}$}
\autonodetext{v}{$\scalemath{1}{1}$}

\edge{u}{a}
\edge{a}{b}
\edge{b}{c}
\edge{b}{d}
\edge{c}{d}
\edge{d}{v}

\end{graph}
}
\vspace{0.6cm}
\caption{The arc gadget}\label{fig:diredgeNIMGRM}
\end{center}
\end{figure}

We show by induction on $|V(G)|$ that for each vertex $u\in V(G)$, $o^+((G,u))=o^-((G',X_u,w))$. 
If $|V(G)|=1$ then $o^+((G,u))=\Pc$. The graph $G'$ is also reduced to a unique vertex $X_0$. The first player has to finish the game by taking the only token on $X_u$.
Hence $o^-((G',X_u,w))=\Pc$.

Now assume $|V(G)|\geqslant 2$. First assume $o^+((G,u))=\Nc$. There is a vertex $v \in V(G)$ such that moving toward $v$ is winning. Let $\widehat{G}$
be the subgraph induced by $V(G)\setminus\{u\}$. We have $o^+((\widehat G,v))=\Pc$. We prove that the first player wins $(G',X_u,w_{G'})$ with the move $(X_u,0,a_{uv})$.
After such a move, the second player is forced to play $(a,0,b)$ and the first player answers with $(b,0,c)$. Once again the second player has no choice, he plays $(c,0,d)$.
The first player plays $(d,0,X_v)$, then the second player has to play in a graph $\tilde{G}$. 
Note that if the first player had played $(b,0,d)$, she would have lost. 
In other words the gadget works as an arc. The player who goes inside the gadget is not the one who goes outside. This shows that playing a move of the form $(X_w,0,a_{wu})$ 
is always losing  because your opponent will have to start one of his further move on $X_u$ and $w(X_u)$ is now equal to $0$. 
Playing a move of the form $(X_w,0,d_{uw})$ would also be losing as we prove in the second part.
Hence, we have $o^-((\widehat G ',X_v,w_{\widehat G '}))=o^-((\tilde G,X_v,w_{\tilde G }))$ and 
by induction hypothesis $o^+((\widehat G ,v))=o^-((\widehat G ',X_v,w_{\widehat G '}))=\Pc$. Therefore $(X_u,0,a)$ is winning and $o^-((G ',X_u,w_{G'} ))=\Nc$.

Reciprocally, assume that $o^-((G',X_u,w_{G'}))=\Nc$. There must exist a winning move. We first show that this move cannot be of the kind $(X_u,0,d_{vu})$. 
In other words, we show that our gadget is oriented. We note that the status of the move $(a_{vu},0,X_v)$ does not depend of the 
moves which will be played in the gadget before we reach $a_{vu}$. In fact, since $w(X_u)$ is now equal to $0$, the players will not be able to come back in the gadget after they get out of it.
We can therefore work case by case to show that $(X_u,0,d_{vu})$ is a losing move.

If the move $(a_{vu},0,X_v)$ is a losing move, the second player wins with the move $(d_{vu},0,c_{vu})$. In fact, the first player has to play $(c_{vu},0,b_{vu})$ and he answers with $(b_{vu},0,a_{vu})$. 
Finally, the first player has to play the losing move, $(a_{vu},0,X_v)$.

On the other hand, if the move $(a_{vu},0,X_v)$ is a winning move, the second player wins with the move $(d_{vu},1,b_{vu})$. 
There are now three possibilities for the first player.
She can answer with $(b_{vu},0,d_{vu})$. Then the second player plays $(d_{vu},0,c_{vu})$ and she has to play $(c_{vu},0,b_{vu})$. In that case, the second player wins since we are under mis\`ere convention and he is on a null weight vertex.
 If she chooses to play $(b_{vu},0,a_{vu})$, then the second player can play the winning move $(a_{vu},0,X_v)$. Finally,
 if she plays $(b_{vu},0,c_{vu})$, the second player answers with $(c_{vu},0,d_{vu})$ and she loses because she is now surrounded by null weight vertices.

Since there is no winning move of the kind $(X_u,0,d_{vu})$, there must be one winning move of the form $(X_u,0,a_{uz})$.
Let $\widehat{G}$ be the subgraph induced by $V(G)\setminus\{u\}$. We focus on the moves following $(X_u,0,a_{uz})$. The second player has no choice and plays $(a_{uz},0,b_{uz})$.
Then the first player has two choices. She can move to the vertex $d_{uz}$. But in this case, the second player will win with $(d_{uz},0,c_{uz})$. So we can assume she
rather plays to $c_{uz}$. Her opponent has no choice and move to $d_{uz}$. Once again she has two choices. The move $(d_{uz},1,X_z)$ is losing because the second player can answer
with $(X_{w},0,d_{uz})$. Hence we can suppose she plays $(d_{uz},0,X_z)$. Since there is no more token on $X_u$ and $X_z$, 
the second player has to play in a graph whose outcome is the same as $(\widehat G',X_z,w_{\widehat G'})$. The outcome of this game is $\Pc$ because $(X_u,0,a_{uz})$ is a winning move. 
By induction hypothesis, $o^+((\widehat G,z))=\Pc$. So moving to $z$ is a winning move in $G$ and $o^+((G,u))=\Nc$.
\end{proof}
As in the case of the normal version of \textsc{NimG-MR} (see \cite{Burke}), the reduction works even if we restrict ourselves to weight functions bounded by $2$. 
\begin{corollary}
 The mis\`ere version of \textsc{NimG-RM+NL} with weight function bounded by $2$ is \textsc{pspace}-complete. 
\end{corollary}
\begin{proof}
 In this case the length of a game never exceeds $2\times |V(G)|$ moves.
Hence the game is in \textsc{pspace}.
\end{proof}
As recalled in Corollary~\ref{cor:dirgeo} and \ref{cor:undirgeo}, \textsc{Vertex Geography} is \textsc{pspace}-complete even on planar directed graphs with maximum degree $3$.
Our gadget has both properties. Furthermore, the reduction does not increase the degree of the original graph vertices, so we get the following result.
\begin{corollary}
 The mis\`ere version of \textsc{NimG-RM+NL} with weight function bounded by $2$ is \textsc{pspace}-complete even when restricted to planar graphs with degree at most $3$. 
\end{corollary}
The previous reduction raises up two questions. Does \textsc{NimG-RM+NL} remain \textsc{pspace}-Hard if we bound the weight function by $1$? And is it still \textsc{pspace}-Hard
if we only consider bipartite graphs? 
In fact, the odd cycle and the vertex with two tokens seem essential to perform our reduction to \textsc{NimG-RM+NL}. 
The two results below show that they are really necessary.

\begin{theorem}\label{th : NGRM1}
 Let $G$ be a graph and $w$ its weight function. If $w$ is constant equal to $1$, we can compute $o^-((G,u,w))$ and find an effective winning strategy in time
 $\mathcal{O}((\sqrt{(|V(G)|}.|E(G)|)$.
\end{theorem}
\begin{proof}
If we allow only one token by vertices, the mis\`ere version of \textsc{NimG-RM+NL} is exactly the same as the normal version of \textsc{Vertex Geography} on undirected graphs. 
As recalled in the table of Section~\ref{sectionGeo}, this game is solvable in polynomial time. Therefore, with only one token allowed, the mis\`ere version of \textsc{NimG-RM+NL} is also solvable in polynomial time. 
\end{proof}
In the next theorem and its corollary, we prove that the problem is also solvable in polynomial time when we play on bipartite graphs. We will assume there is no null weight vertex 
at the beginning. 
It is not really a restriction, since in mis\`ere version, the outcome of the game is the same if we played on a graph $G$ or on the subgraph of $G$ induced by the vertices 
with at least one token.  

\begin{theorem}
 Let $G$ be a bipartite graph and $w$ a strictly positive weight function. The position $(G,u,w)$ of \textsc{NimG-RM+NL} is winning in mis\`ere convention if and only if all the maximum
matchings of $G$ cover the starting vertex $u$.
\end{theorem}
\begin{proof}
 Since $G$ is bipartite, we can split $V(G)$ into two disconnected subsets $L$ and $R$ such that $u\in L$. Note that the first player will always remove 
tokens from vertices of $L$ whereas his opponent will remove tokens from $R$. 
 
 We now assume that all the maximum matchings of $G$ cover $u$. Let $M$ be such a matching.
We show that removing all the heaps on the current vertex and then moving along an edge of $M$ is a winning strategy for the first player.
We look at the first time she cannot follow the above strategy. There are two possibilities. Firstly, she is on a vertex with no token on it. In this case,
she wins and the strategy is indeed a winning one. The second possibility is that she is on a vertex not visited before but uncovered by $M$. We show
this possibility never happens.
In that case we have a list of edges $(f_1,s_1,...,f_n,s_n)$ such that the $f_i$ stands for the edges followed by the first player and 
the $s_i$ the edges followed by the second one.
Since the game is not already ended, the $s_i$ have no vertex in $L$ in common. Otherwise, the first player would have been on a null weight vertex. 
Hence the $f_i$ which are in $M$ are all distinct and they have no vertex in common either. In other words $(f_1,s_1,...,f_n,s_n)$ is a path. So 
$(M\cup\{s_1,...,s_n\})\setminus\{f_1,...,f_n\}$ is a maximum matching which does not cover $u$. This is in contradiction with our hypothesis. 

Reciprocally, assume there is a maximum matching $M$ which does not cover $u$. Let $v$ be the first vertex toward which the first player moves. Since
$M$ is maximum, $v$ is covered by $M$. Hence the second player can follow the same strategy as we saw above. Showing that this strategy is winning for him can be done as 
before. The case where he is stuck on a vertex uncovered by $M$ will not appear either. In this case, it actually leads to an augmenting path 
$(f_1,s_1,...,f_n,s_n,f_{n+1})$, which would contradict the maximality of $M$.
\end{proof}

\begin{corollary}
 The mis\`ere version of the game \textsc{NimG-RM+NL} is solvable in polynomial time on bipartite graph. Computing $o^-((G,u,w))$ and finding an effective winning strategy can be done in time 
$\mathcal{O}(\sqrt{|V(G)|}.|E(G)|)$.
\end{corollary}
\begin{proof}
 Let $G'$ be the subgraph induced by $V(G)\setminus\{u\}$. We compute $C$ the cardinal of a maximum matching of $G$, then we compute $C'$ the cardinal of a maximum matching of $G'$.
Both of these operations can be done in time $\mathcal{O}(\sqrt{|V(G)|}.|E(G)|)$, using the Edmond-Karp's algorithm. If $C=C'$, there is a maximum matching of $G$ which does not cover $u$, then
$o^-((G,u,w))=\Pc$. On the contrary, if $C'<C$, all the maximum matchings of $G$ cover $u$ and $o^-((G,u,w))=\Nc$. Moreover, the Edmond-Karp's algorithm gives us a 
maximum matching of $G$ covering $u$. The effective winning strategy is as follows: take all the tokens on the current vertex, then move along the edge of the matching. 
\end{proof}
To finish with the game \textsc{NimG-RM}, we give an algorithm for the mis\`ere version of \textsc{NimG-RM+L}. One more time, we suppose there is no null weight vertex at the beginning. 
We already saw it does not matter.
\begin{theorem}
 Let $G$ be a graph with a loop on all its vertices and $w$ its weight function. Computing $o^-((G,u,w))$ and finding an effective winning strategy can be done in time 
$\mathcal{O}(\sqrt{|V(G)|}.|E(G)|)$.
\end{theorem}
\begin{proof}
 Let $T$ be the subset of $V(G)$ defined by $T=\{u\in V(G)~|~w(u)\geq 2\}$. We show that for any $u$ in $T$, $(G,u,w)$ is a winning position. 
Assume there is a winning move of the form $(u,0,v)$ with $v\not = u$. In this case $(G,u,w)$ is clearly winning. If all the moves of this kind are losing moves, 
the first player decreases $w(u)$ to $1$ and then uses the loop to stay on $u$. His opponent will have to play one of the losing moves $(u,0,v)$, so $(G,u,w)$ is also 
winning in this case.

Now, let $u$ be a vertex with $w(u)=1$ and let $C_u$ be the connected component of $G\setminus T$ which contains $u$. Since moving outside $C_u$ is always a losing move,
we have $o^-((G,u,w))=o^-((C_u,u,w))$. Using theorem \ref{th : NGRM1} we can compute $o^-((C_u,u,w)$ in the expected time.
\end{proof}

We conclude this section with the result for \textsc{NimG-MR} in mis\`ere convention. 
For our result, the loops do not matter, so we forget the \textsc{+L} and \textsc{+NL}. Burke and George 
only proved in \cite{Burke} that \textsc{NimG-MR+L} is \textsc{pspace}-Hard in normal convention. 
Carefully looking at their proof, it turns out that it works the same for \textsc{NimG-MR+NL} in normal convention. To get our result, we reduce the normal version 
of \textsc{NimG-MR} to its mis\`ere version.
\begin{theorem}
 The game \textsc{NimG-MR+L} and the game \textsc{NimG-MR+NL} are \textsc{pspace}-Hard in mis\`ere convention.
\end{theorem}
\begin{proof}
 Let G be a graph. We construct the graph $G'$ by adding to each vertex of $G$ a chain of three vertices with weight $1$ (see figure \ref{fig : reduNMR}).
We claim that $o^+(G)=o^-(G')$, because when you play on $G'$, moving outside $G$ is always a losing move. 
The details of the proof are similar to theorem \ref{thm:dirgeo}.
\end{proof}
\begin{figure}
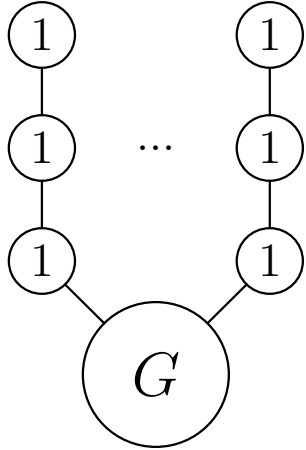

\begin{center}
 \scalemath{1.5}{
\begin{graph}(3,4)
\graphnodesize{0.6}
\graphnodecolour{1}
\autodistance{1}
\roundnode{g}(1,0)[\graphnodesize{1.3}]
\autonodetext{g}{$\scalemath{1.3}{G}$}
\graphnodesize{0.6}
\graphnodecolour{1}
\roundnode{a}(0,1)
\autonodetext{a}{$\scalemath{1}{1}$}
\roundnode{b}(0,2)
\autonodetext{b}{$\scalemath{1}{1}$}
\roundnode{c}(0,3)
\autonodetext{c}{$\scalemath{1}{1}$}
\freetext(1,2){$\scalemath{1}{...}$}
\roundnode{a2}(2,1)
\autonodetext{a2}{$\scalemath{1}{1}$}
\roundnode{b2}(2,2)
\autonodetext{b2}{$\scalemath{1}{1}$}
\roundnode{c2}(2,3)
\autonodetext{c2}{$\scalemath{1}{1}$}
\edge{g}{a}
\edge{a}{b}
\edge{b}{c}
\edge{g}{a2}
\edge{a2}{b2}
\edge{b2}{c2}
\end{graph}
}
\vspace{0.8cm}
\caption{Reduction from the normal version to the mis\`ere one}\label{fig : reduNMR}
\end{center}
\end{figure}

\section{Conclusion.}

In this work, we made a comprehensive study of the complexity of the mis\`ere version of the three games \textsc{Geography}, \textsc{NimG-RM} 
and \textsc{NimG-MR}. Except for the variant of \textsc{NimG-RM} with a loop on each vertex, all these games are \textsc{pspace}-hard or complete on general graphs.
This shows that the mis\`ere versions of those games are never easier than the normal ones. 
For \textsc{NimG-RM}+NL and \textsc{Vertex Geography} on undirected graphs there is even a real gap between the normal and the mis\`ere version,
since an effective winning strategy can be computed in polynomial time under normal play.

Our reductions for \textsc{Vertex Geography} on undirected graphs, \textsc{NimG-RM} 
and \textsc{NimG-MR} made an intensive use of odd cycles. Hence we investigated the restriction of those games to bipartite graphs. 
For \textsc{NimG-RM}, we showed that the game becomes polynomial in that case, whereas for \textsc{Vertex Geography} and \textsc{NimG-MR} the complexity is still
unknown.
\section{Acknowledgements.}

We thank Sylvain Gravier for the helpful discussions we had together all along this work.

\bibliographystyle{plain}

\bibliography{bibliothese}

\end{document}